\documentclass[conference]{IEEEtran}

\IEEEoverridecommandlockouts
\usepackage{cite}
\usepackage{amsmath,amssymb,amsfonts}
\usepackage{algorithmic}
\usepackage{graphicx}
\usepackage{textcomp}
\usepackage{xcolor}

\usepackage{amsthm}
\usepackage{mathtools}
\usepackage{verbatim}

\newtheorem{myTheor}{Theorem}

\newcounter{MYtempeqncnt}

\newcommand{\erfc}{\mathrm{erfc}}

\def\BibTeX{{\rm B\kern-.05em{\sc i\kern-.025em b}\kern-.08em
    T\kern-.1667em\lower.7ex\hbox{E}\kern-.125emX}}
\begin{document}

\title{Lower Bound on the Error Rate of Genie-Aided Lattice Decoding
\thanks{This work was supported by JSPS Kakenhi Grant Number JP 19H02137 and 21H04873.}
}

\author{\IEEEauthorblockN{Jiajie Xue and Brian M.~Kurkoski}
\IEEEauthorblockA{Japan Advanced Institute of Science and Technology\\
1-1 Asahadai, Nomi, Ishikawa, Japan \\
\{xue.jiajie, kurkoski\}@jaist.ac.jp}
}

\maketitle

\begin{abstract}
A genie-aided decoder for finite dimensional lattice codes is considered. The decoder may exhaustively search through all possible scaling factors $\alpha \in \mathbb{R}$.  We show that this decoder can achieve lower word error rate (WER) than the one-shot decoder using $\alpha_{MMSE}$ as a scaling factor. 
A lower bound on the WER for the decoder is found by considering the covering sphere of the lattice Voronoi region. The proposed decoder and the bound are valid for both power-constrained lattice codes and  lattices.
If the genie is applied at the decoder, E8 lattice code has 0.5 dB gain and BW16 lattice code has 0.4 dB gain at WER of $10^{-4}$ compared with the one-shot decoder using $\alpha_{MMSE}$. A method for estimating the WER of the  decoder is provided by considering the effective sphere of the lattice Voronoi region, which shows an accurate estimate for E8 and BW16 lattice codes. In the case of per-dimension power $P \rightarrow \infty$, an asymptotic expression of the bound is given in a closed form. A practical implementation of a simplified decoder is given by considering CRC-embedded $n=128$ polar code lattice. 
\end{abstract}


\section{Introduction}

An $n$-dimensional lattice is a discrete additive subgroup of real number space $\mathbb{R}^n$. Lattices obtain coding gain by providing Euclidean distance between codewords; lattices share the same algebra as physical channels. By those properties, lattices can be considered as a type of channel coding scheme. Lattices are infinite constellation and are not suitable for power-constrained wireless communications. By constructing nested lattice codes $\Lambda_c/ \Lambda_s$ consisting of a coding lattice $\Lambda_c$ and a shaping lattice $\Lambda_s$, lattice codes can satisfy power constraints for wireless communications. Lattice codes can achieve additive white Gaussian noise (AWGN) channel capacity by applying minimum-distance decoding with respect to the lattice code \cite{urbanke1998lattice}. More significantly, Erez and Zamir proved that the channel capacity can also be achieved by instead applying lattice decoding with an minimum mean square error (MMSE) scaling factor $\alpha_{MMSE}$ \cite{erez2004achieving}.
Practical designs of lattice codes are also being studied. Construction D/D' lattices use binary codes, such as polar codes \cite{liu2018construction} and LDPC codes \cite{zhou2022construction}, which show excellent WER performance on the power-constrained AWGN channel. Nested lattice codes are also component codes for compute-and-forward relaying \cite{nazer2011compute}, which can provide high network throughput by applying network coding.

Lattice decoding applying $\alpha_{MMSE}$ is optimal when $n \rightarrow \infty$ \cite{erez2004achieving}. However, $\alpha_{MMSE}$ may not be optimal for finite dimensional lattice decoding. In this paper, we propose a decoder for finite dimensional lattice codes which performs a genie-aided exhaustive search through all possible scaling factors $\alpha \in \mathbb{R}$. The genie informs the decoder whether the decoding result is correct or not. If the genie says the current decoding result is not correct, then the decoder will choose a different $\alpha$ and re-try decoding. 
Because all $\alpha \in \mathbb{R}$ can be used for decoding, the decoder is expected to achieve a lower WER than one-shot decoding using only $\alpha_{MMSE}$. 

As the main result of this paper, for $n$-dimensional lattice codes, we derive a lower bound on the WER for the proposed decoder. The decoder and the bound are valid for both power-constrained lattice codes and  lattices. 
The lower bound on WER is found by considering the \emph{covering sphere} of the lattice Voronoi region and an estimate of the WER is found by considering the \emph{effective sphere} of the lattice Voronoi region. The estimate is accurate for E8 and BW16 lattice codes which have sphere-like Voronoi regions. An asymptotic expression for the lower bound as the per-dimensional power $P$, or equivalently code rate $R$, goes to infinity is given in closed form. 
A comparison among WER obtained by the error bound, the genie-aided exhaustive search lattice decoder and conventional lattice decoders are provided using E8 and BW16 lattice codes.
As a result, if the genie is applied at the decoder, the E8 lattice code has 0.5 dB gain and the BW16 lattice code has 0.4 dB gain at WER of $10^{-4}$, compared with the one-shot decoder using $\alpha_{MMSE}$. For practical implementation, the genie can be implemented by error detection using a cyclic redundancy check (CRC) code. Numerical evaluation using a CRC-embedded polar code lattice with $n=128$ is provided where the decoder uses three $\alpha$ candidates instead of exhaustive search. About 0.1 dB gain at WER of $10^{-4}$ can be obtained if the decoder is allowed to re-try decoding.

For the fonts of variables in this paper, scalars are denoted as \emph{italic}, e.g.\ $r$, vectors are denoted as lower-case bold, e.g.\ $\mathbf{g}$, and matrices using upper-case bold, e.g.\ $\mathbf{G}$. We use $n$ for the lattice dimension. The meaning of $n$ keeps when it appears at subscript or superscript, such as the $n$-dimensional real number space $\mathbb{R}^n$. 

\section{Preliminaries}

This section gives definitions related to lattices codes and encoding/decoding scheme. The universal lower bound on error probability for lattice codes is also introduced.

\subsection{Lattices and lattice codes} \label{sec_II_A}

Let $\mathbf{g}_1, \mathbf{g}_2,\ldots, \mathbf{g}_n \in \mathbb{R}^n$ be $n$ linearly independent column vectors. A $n$-dimensional lattice $\Lambda$ can be formed by:
\begin{equation*}
    \Lambda= \{ \mathbf{Gb} : \mathbf b \in \mathbb Z^n \}
\end{equation*}
where $\mathbf{G}= [\mathbf{g}_1, \mathbf{g}_2, ... \mathbf{g}_n] \in \mathbb{R}^{n\times n}$ is a generator matrix. 
The lattice quantizer $Q_{\Lambda}(\mathbf{y})$ for $\Lambda$ finds the nearest lattice point $\mathbf{x} \in \Lambda$ for arbitrary $\mathbf y \in \mathbb{R}^n$ and can be expressed as:
\begin{equation*}
    Q_{\Lambda}(\mathbf{y})= \mathop{\arg\min}_{\mathbf x \in \Lambda} \|\mathbf{y}- \mathbf{x}\|^2.
\end{equation*}
Using the lattice quantizer $Q_{\Lambda}(\mathbf{y})$, the modulo operation for lattice $\Lambda$ can be defined by:
\begin{equation*}
    \mathbf{y} \bmod \Lambda = \mathbf{y}- Q_{\Lambda}(\mathbf{y}).
\end{equation*}

The Voronoi region $\mathcal{V}$ of lattice point $\mathbf{x} \in \Lambda$ is the set of $\mathbf y$ such that $\mathbf{y}$ is closer to $\mathbf{x}$ than any other lattice point. The volume of the Voronoi region of the lattice is $V_n= |\det(\mathbf{G})|$. 
The \emph{covering sphere} $\mathcal{S}_c$ with radius $r_c$ is the sphere of minimal radius that can cover the whole Voronoi region $\mathcal{V}$, i.e.\ $\mathcal{V} \subseteq \mathcal{S}_c$. The \emph{effective sphere} $\mathcal{S}_e$ with radius $r_e$ is the sphere with volume $V(\mathcal{S}_e)$ which is equal to the volume of the Voronoi region $V_n$, i.e.\ $V(\mathcal{S}_e) = V_n$. The volume of a sphere $\mathcal{S}$ in $n$-dimensional space with radius $r$ is $V(\mathcal{S})= \frac{\pi^{n/ 2} r^n}{\Gamma(\frac{n}{2}+ 1)}$, where $\Gamma(\cdot)$ is gamma function. The relationship among $\mathcal{V}$, $\mathcal{S}_c$ and $\mathcal{S}_e$ with $n=2$ is shown in Fig.~\ref{fig_eff_cov_sphere}. More details related to covering sphere and effective sphere can be found in \cite[Sec.\ 2]{erez2005lattices}.

A nested lattice code can be used as power-constrained lattice code for  wireless communications. 
Let two lattices $\Lambda_c$ and $\Lambda_s$ satisfy $\Lambda_s \subseteq \Lambda_c$ and form a quotient group $\Lambda_c / \Lambda_s$. A nested lattice code $\mathcal{C}$ is the set of coset leaders of the quotient group $\Lambda_c / \Lambda_s$ or equivalently $\mathcal{C}= \{\boldsymbol{x} \bmod \Lambda_s:\boldsymbol{x} \in \Lambda_c\}$.
The fine lattice $\Lambda_c$ is called the coding lattice while coarse lattice $\Lambda_s$ is called the shaping lattice.

\subsection{Encoding/decoding scheme}

The encoding/decoding scheme used in this paper is shown in Fig.~\ref{fig_channel_model_awgn} and real number transmission is assumed. Let an $n$-dimensional nested lattice code be $\mathcal{C}= \Lambda_c/\Lambda_s$ and encode the $n$-dimensional message $\mathbf{b}$ into the coded message $\mathbf{x} \in \mathcal{C}$. $\mathbf{x}$ is transmitted over the AWGN channel. The channel output is $\mathbf{y}= \mathbf{x}+ \mathbf{z}$, where $\mathbf{z} \sim \mathcal{N}(0, \sigma^2 \mathbf I_n)$ is the i.i.d.\ Gaussian noise and $\mathbf I_n$ is the identity matrix. The probability density function of $n$-dimension Gaussian noise $\mathbf{z}$ is given by $g_n(\mathbf{z)}$:
\begin{equation*}
    g_n(\mathbf{z})= \frac{1}{(2 \pi \sigma^2)^{n/2}} e^{-\|\mathbf{z}\|^2/2 \sigma^2}
\end{equation*}
The average power the codebook $\mathcal{C}$ is denoted as $P$. The signal-to-noise ratio (SNR) for the AWGN channel is $SNR= P/\sigma^2$. At the receiver, the decoder first scales the channel output by $\alpha \in \mathbb{R}$ then uses the lattice quantizer and modulo to estimate the original information $\hat{\mathbf{b}}$. The decoding succeeds if $\alpha \mathbf{y}$ falls into the Voronoi region of $\mathbf{x}$. The MMSE scaling factor $\alpha_{MMSE}$ introduced in \cite{erez2004achieving} is considered as the asymptotically optimal factor and is $\alpha_{MMSE}= \frac{SNR}{1+ SNR}$.
For convenience, we will denote the whole decoding process as $\hat{\mathbf{b}}= Dec(\alpha \mathbf{y)}$.

\begin{figure}[t]
    \centering
    \includegraphics[scale=0.4]{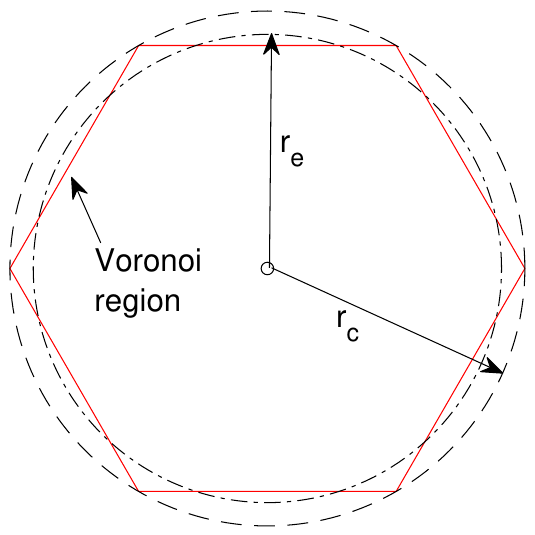}
    \caption{Relationship among lattice Voronoi region, covering sphere and effective sphere.}
    \label{fig_eff_cov_sphere}
\end{figure}

\begin{figure}[t]
    \centering
    \includegraphics[scale=0.25]{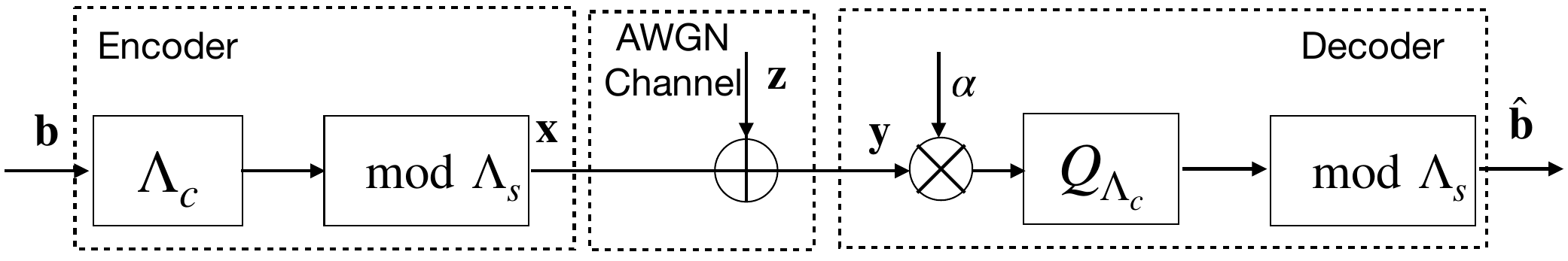}
    \caption{Encoding/decoding scheme and channel model.}
    \label{fig_channel_model_awgn}
\end{figure}

\subsection{Universal error bound for lattice codes}

A universal bound on decoding error probability for lattices using a sphere in $n$-dimensional space was given in \cite{tarokh1999universal}. Given effective radius $r_e$ and noise variance $\sigma^2$, the error probability of any $n$-dimensional lattice is lower bounded by:
\begin{align} \label{equ_universal_even}
    P_e \geq e^{-t} \left(1+ \frac{t}{1!}+ \dots+  \frac{t^{n/2- 1}}{(n/2- 1)!}\right)
\end{align}
for $n$ even; while for odd $n$:
\begin{align} \label{equ_universal_odd}
    P_e \geq& \erfc(t^{\frac{1}{2}})\\ \notag
    & + e^{-t} \left(\frac{t^{1/2}}{(1/2)!}+ \frac{t^{3/2}}{(3/2)!}+ \dots+  \frac{t^{n/2- 1}}{(n/2- 1)!}\right)
\end{align}
where $t= r_e^2/ 2 \sigma^2$. 
The bound on probability of correct decoding $P_c= 1- P_e$ is the probability of $n$-dimensional Gaussian noise falls into the effective sphere, which can be computed by closed form using \eqref{equ_universal_even} and \eqref{equ_universal_odd}. 

\section{Lower bound on error rate}

In this section, a genie-aided exhaustive search decoder is proposed. The lower bound on the WER for the decoder is derived by considering the covering sphere. An estimate on the achievable WER is found by considering the effective sphere. The asymptotic expression of the lower bound is given for the case of per-dimension power $P \rightarrow \infty$.

\subsection{Decoding strategy}

Instead of only use $\alpha_{MMSE}$ at decoding, a genie-aided exhaustive search decoder is proposed to achieve lower WER. 
The genie-aided decoder attempts to exhaustively search through all $\alpha \in \mathbb{R}$ to find some valid $\alpha$ that can correctly decode the received message. 
Given received message $\mathbf{y} \in \mathbb{R}^n$ and some $\alpha \in \mathbb{R}$, the decoder will  decode  $\hat{\mathbf{b}}= Dec(\alpha \mathbf{y)}$. Then the genie will check whether $\hat{\mathbf{b}}= \mathbf{b}$. If $\hat{\mathbf{b}} = \mathbf{b}$, the decoder will output $\hat{\mathbf{b}}$ as the decoding result. If $\hat{\mathbf{b}} \neq \mathbf{b}$, the decoder will select different $\alpha$ and re-try decoding. For correctly decodable $\mathbf{y}$, decoder can always find some $\alpha^*$, by which $\hat{\mathbf{b}}^*= Dec(\alpha^* \mathbf{y})= \mathbf{b}$; otherwise such $\alpha$ doesn't exist and decoding fails. Since the search range of $\alpha$ is the whole real number space $\mathbb{R}$ rather than one single value $\alpha_{MMSE}$, the proposed decoder is expected to achieve lower WER than the one-shot decoder.

Based upon the construction of the decoder, the received message $\mathbf{y}$ is correctly decodable if and only if a line connecting $\mathbf{y}$ and the original point $\mathbf{0}$ passes through the Voronoi region of $\mathbf{x}$. Furthermore, we can define the decodable region as the set of $\mathbf{y}$ such that there exists $\alpha \in \mathbb{R}$ and $\alpha \mathbf{y} \in \mathcal{V}$. Since the decodable region for transmitted message $\mathbf{x}$ depends on the message itself, the decodable region formed by the Voronoi region can be denoted using $\mathcal{D}(\mathbf{x})$:
\begin{align*} \label{equ_def_dec_region}
    \mathcal{D}(\mathbf{x})= \{\mathbf{y} \mid \exists\alpha \in \mathbb{R},\ \alpha \mathbf{y} \in \mathcal{V}(\mathbf{x})\}
\end{align*}
where $\mathcal{V}(\mathbf{x})$ is Voronoi region of given $\mathbf{x}$. The decodable region of transmitted message $\mathbf{x}$ forms an $n$-dimensional cone-like region. 
Fig.~\ref{fig_possible_region_A2} shows an example of the cone-like decodable region using A2 lattice. Given transmitted message $\mathbf{x}$ and the corresponding Voronoi region $\mathcal{V}(\mathbf{x})$, we can draw two lines connecting $\mathbf{0}$ and the corners of $\mathcal{V}(\mathbf{x})$ to form the decodable region $\mathcal{D}(\mathbf{x})$ by which the Voronoi region $\mathcal{V}(\mathbf{x})$ should be completely covered. By the example of Fig.~\ref{fig_possible_region_A2}, the received message is correctly decodable if $\mathbf{y} \in \mathcal{D}(\mathbf{x})$, such as $\mathbf{y}_1$, because there exists some $\alpha$ that the scaled received message $\alpha \mathbf{y}$ falls into $\mathcal{V}(\mathbf{x})$; while for $\mathbf{y}_2$, such $\alpha$ does not exist. 


\begin{figure}[t]
    \centering
    \includegraphics[scale=0.5]{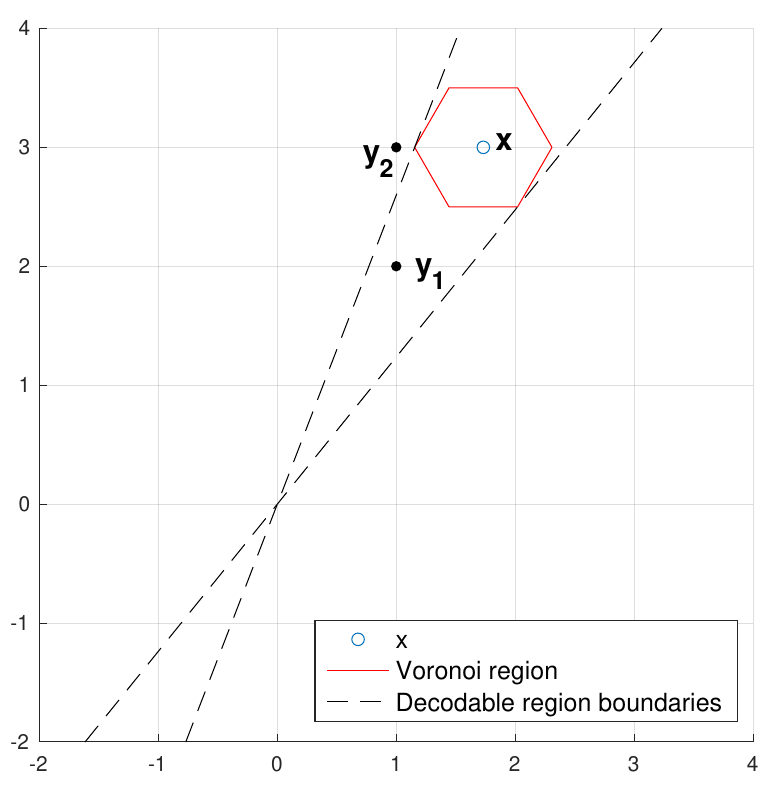}
    \caption{Example of decodable region $\mathcal{D}(\mathbf{x})$ using A2 lattice and $\mathbf{x}= (\sqrt{3}, 3)$.}
    \label{fig_possible_region_A2}
\end{figure}

\subsection{Lower bound on error rate using covering sphere} \label{sec_III_2}

Now we give a lower bound on the WER for the genie-aided exhaustive search decoder by computing the probability that $\mathbf{y}$ falls into a superset of the decodable region $\mathcal{D}(\mathbf{x})$. The shape of $\mathcal{D}(\mathbf{x})$ depends on $\mathbf{x}$ and is hard to find in general. Instead of forming $\mathcal{D}(\mathbf{x})$ using $\mathcal V(\mathbf{x})$, we define a new decodable region $\mathcal{D}_c$ using the covering sphere $\mathcal{S}_c$ centered at $\mathbf{x}$ as $\mathcal{D}_c= \{\mathbf{y} \mid \exists\alpha \in \mathbb{R},\ \alpha \mathbf{y} \in \mathcal{S}_c\}$. Given lattice codes $\mathcal{C}$ and the message power $\|\mathbf{x}\|^2$, the $\mathcal{D}_c$ only depends on $P_{\mathbf x}$ 
and $r_c$ which are constant, where $P_{\mathbf x}= \|\mathbf{x}\|^2/n$.

Let $P_{e,Dec}$ be the probability of word error for given $\mathbf{x}$:
\begin{equation} \label{equ_errorP_dec}
    P_{e,Dec}= 1- \rm{Prob}(\mathbf{y} \in \mathcal{D}(\mathbf{x}))
\end{equation}
where $\mathbf{y}= \mathbf{x}+ \mathbf{z}$. Similarly, we can define $P_{e,cover}$ when $\mathcal{S}_c$ and $\mathcal{D}_c$ are considered:
\begin{equation} \label{equ_errorP_cover}
    P_{e,cover}= 1- \rm{Prob}(\mathbf{y} \in \mathcal{D}_c)
\end{equation}
For finite dimensional lattice codes, the Voronoi region is not a sphere and $\mathcal{V} \subset \mathcal{S}_c$ holds. Therefore $\mathcal{D}(\mathbf{x}) \subset \mathcal{D}_c$ also holds for any $\mathbf{x} \in \mathcal{C}$. As a result, the WER of the decoder is lower bounded by $P_{e,cover}$:
\begin{equation} \label{equ_lower_bound_def}
    P_{e,Dec} > P_{e,cover}= 1- \rm{Prob}(\mathbf{y} \in \mathcal{D}_c)
\end{equation}
The inequality in \eqref{equ_lower_bound_def} is strict for finite dimensional lattice codes and becomes increasingly tight as the Voronoi region becomes sphere-like.

By obtaining an exact expression for $P_{e,cover}$, we can give a lower bound on the WER for the genie-aided exhaustive search decoder by using \eqref{equ_lower_bound_def}, in as Theorem~\ref{them_1} below. Note that the decoder of Erez and Zamir's uses a lattice decoder which only finds the closest lattice point while ignoring the boundary of the codebook. The proposed lattice decoder and its lower bound are valid for both power-constrained lattice codes $\mathcal{C}$ and lattices $\Lambda$. Without loss of generality, Theorem~\ref{them_1} is given by considering $\Lambda$.
\begin{myTheor} \label{them_1}
\rm Let non-zero $\mathbf x$ be a lattice point of an $n \geq 2$ dimensional lattice $\Lambda$ having covering radius $r_c$ and per-dimensional power $P_{\mathbf{x}}= \|\mathbf{x}\|^2/n$. With the restriction of $r_c^2 < n P_{\mathbf x}$, the probability of word error for the genie-aided decoder on the AWGN channel with noise variance $\sigma^2$ is lower bounded by:
\begin{align} \label{equ_sphere_bound}
    P_{e,Dec} > 1- \int_{-\infty}^{\infty} \frac{1}{\sqrt{2 \pi \sigma^ 2}} e^{- \frac{z^2}{2 \sigma^2}} (1 - h(z)) dz
\end{align}
where if $n$ is odd:
\begin{align*}
    h(z) =  e^{- t}\left(\sum_{k= 0}^{(n- 3)/ 2} \frac{t^k}{k!}\right)
\end{align*}
and if $n$ is even:
\begin{align*}
    h(z) =  \erfc(t^ {1/2}) + e^{- t}\left(\sum_{k= 1}^{(n- 2)/ 2} \frac{t^{k- 1/2}}{(k- 1/ 2)!}\right)
\end{align*}
with $t= f^2(z) / (2 \sigma^2)$ and $f(z)= \left| \frac{r_c}{\sqrt{n P_{\mathbf x} - r_c^ 2}} z + \sqrt{\frac{n P_{\mathbf x} r_c^ 2}{n P_{\mathbf x} - r_c^ 2}} \right|$. 
\end{myTheor}

\begin{proof}
The probability that $\mathbf{y}$ falls into $\mathcal{D}_c$ in \eqref{equ_lower_bound_def} is:
\begin{align} \label{equ_lower_bound_noise}
    \rm{Prob}(\mathbf{y} \in \mathcal{D}_c)= \rm{Prob}(\mathbf{z} \in \mathcal{D}_c')= \int_{\mathcal{D}_c'} \it g_n(\mathbf{z})\,  d \mathbf{z}
\end{align}
where $\mathcal{D}_c'= \mathcal{D}_c- \mathbf{x}$. Due to the symmetry of the Gaussian noise and the covering sphere, it is equivalent to consider a rotated coordinate system where the $z_1$ axis corresponds to the line connecting the vertex of the decodable region and $\mathbf x$. In the rotated coordinate system, the vertex of the decodable region is $(-\sqrt{n P_{\mathbf x}},0,0,\ldots,0)$. 
(See Fig.~\ref{fig_possible_noise_region} for an example in 2 dimensions). 
The Euclidean distance $r_z$ between the boundaries of decodable region and $z_1$ axis 
is a function of $z_1$ as 
$r_z= f(z_1)= \left|\frac{r_c}{\sqrt{n P_{\mathbf x}- r_c^ 2}} z_1+ \sqrt{\frac{n P_{\mathbf x} r_c^ 2}{n P_{\mathbf x}- r_c^ 2}} \right|$. 

Due to the independence among dimensions of Gaussian noise, \eqref{equ_lower_bound_def} can be written as: 
\setcounter{equation}{7}
\begin{align} \label{equ_p_c_2}
    P_{e,cover}&= 1- \int_{-\infty}^{\infty} \frac{1}{\sqrt{2 \pi \sigma^ 2}} e^{- \frac{z_1^2}{2 \sigma^2}} P_s(r_z)  d z_1
\end{align}
with
\begin{align} \label{equ_p_c_2_Ps}
    P_s(r_z)= \int_{\mathcal{S}_{n-1}(r_z)} g_{n- 1}(z_2, z_3, \dots, z_n)\, d z_2 \dots d z_n
\end{align}
where $\mathcal{S}_{n-1}(r_z)$ is an $n-1$ dimensional sphere with radius $r_z$. The region $\mathcal{S}_{n-1}(r_z)$ can be considered as a truncated slice of a $n$ dimensional Gaussian which is the intersection of the decodable region $\mathcal{D}_c'$ and a hyperplane orthogonal the $z_1$ axis.
A method to solve the integral $P_s(r_z)$ is given in \cite{tarokh1999universal} and the result for $n-1$ dimensions case is modified as \eqref{equ_p_c_4},
where $t= f^2(z)/ (2 \sigma^2)$. Note that the integration in $P_s(r_z)$ is taken over $n-1$ dimensions, the even and odd are opposite of \eqref{equ_universal_even} and \eqref{equ_universal_odd}. By combining \eqref{equ_p_c_2} and \eqref{equ_p_c_4}, the probability of correct decoding $\rm{Prob}(\mathbf{z} \in \mathcal{D}_c')$ is obtained. Then the lower bound on the WER using the genie-aided exhaustive search decoder follows by $P_{e,Dec} > P_{e,cover}= 1- \rm{Prob}(\mathbf{z} \in \mathcal{D}_c')$.
\end{proof}

The bound is meaningful only when $r_c^2 < n P_{\mathbf x}$ is satisfied, or equivalently $P_{\mathbf x}> r_c^2/ n$. Because if existing $\mathbf{x}' \in \Lambda$ having $P_{x'} \leq r_c^2/ n$, such $\mathbf{x}'$ may not exist for all lattices, the $\mathcal{D}_c$ doesn't form a cone-like region and $P_{e, cover}= 0$. Also, since there is only single integral contained in \eqref{equ_sphere_bound}, the error bound can be computed numerically. 

\subsection{Estimate error rate using effective sphere} \label{sec_III_3}

Instead of a lower bound, an estimate on the achievable WER for the decoder is found by considering the effective sphere $\mathcal{S}_e$ centered at given $\mathbf{x}$ instead of the covering sphere. As introduced in Subsection~\ref{sec_II_A}, the effective sphere $\mathcal{S}_e$ is the sphere having same volume as its corresponding Voronoi region, i.e.\ $V(\mathcal{S}_e)= V_n$. We can use an effective sphere to approximate the Voronoi region of lattice codes. Note this only provides an estimate of the WER but not a bound, because neither $\mathcal{V} \subset \mathcal{S}_e$ nor $\mathcal{S}_e \subset \mathcal{V}$ is satisfied. Denote the decodable region considering effective sphere $\mathcal{S}_e$ as $\mathcal{D}_e= \{\mathbf{y} \mid \exists\alpha \in \mathbb{R},\ \alpha \mathbf{y} \in \mathcal{S}_e\}$ and corresponding probability of word error as $P_{e,effc}= 1- \rm{Prob}(\mathbf{y} \in \mathcal{D}_e)$. The estimated WER $P_{e,effc}$ is obtained by replacing $r_c$ in \eqref{equ_sphere_bound} using the radius $r_e$ of the effective sphere $\mathcal{S}_e$. For lattice codes having sphere-like Voronoi region, it implies that $\rm{Prob}(\mathbf{y} \in \mathcal{D}(\mathbf{x})) \approx \rm{Prob}(\mathbf{y} \in \mathcal{D}_e) < \rm{Prob}(\mathbf{y} \in \mathcal{D}_c)$, hence $P_{e,Dec} \approx P_{e,effc} > P_{e, cover}$. Later, we will give a numerical comparison between $P_{e,Dec}$ and $P_{e,effc}$ using E8 and BW16 lattice codes. We will show that the estimation is accurate for E8 and BW16 lattice codes. We hypothesize this because these lattice codes have sphere-like Voronoi regions. 

\begin{figure}[t]
    \centering
    \includegraphics[scale=0.18]{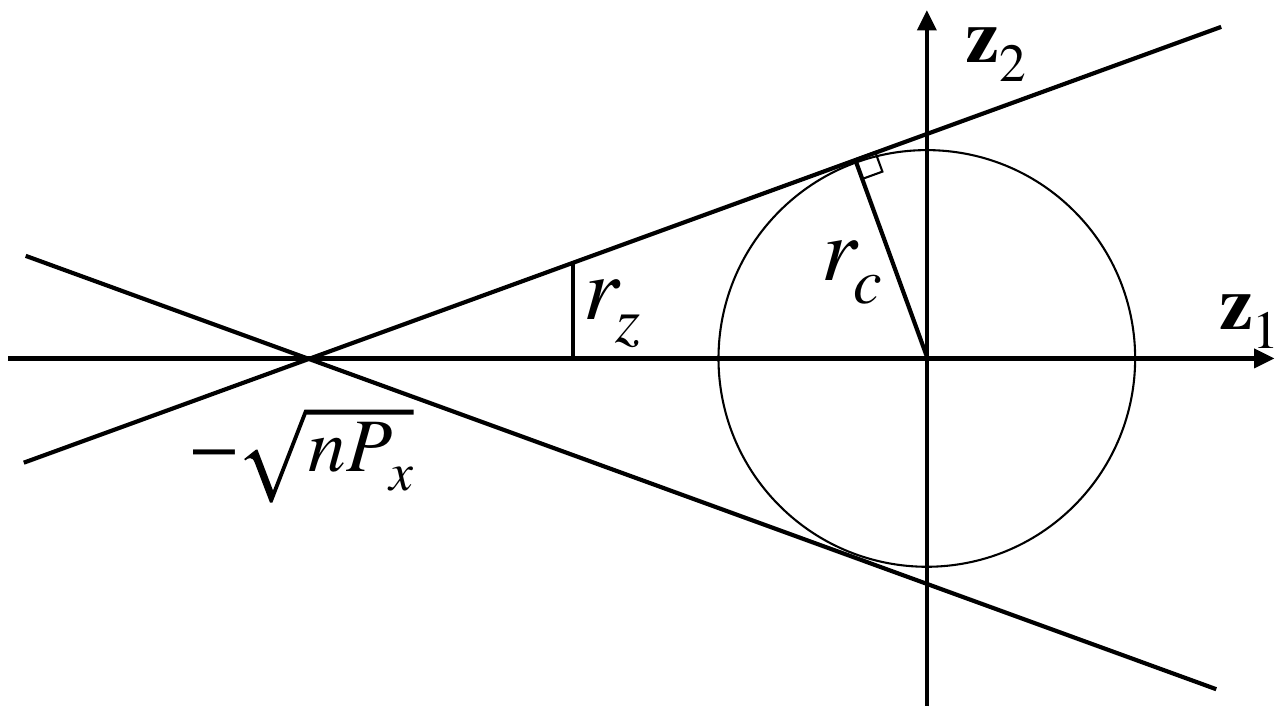}
    \caption{Example of rotated $\mathcal{D}_c'$ in 2 dimensions. The vertex of the decodable region is $(-\sqrt{n P_{\mathbf x}},0)$.}
    \label{fig_possible_noise_region}
\end{figure}

\begin{figure*}[t]
\normalsize
\setcounter{MYtempeqncnt}{\value{equation}}
\setcounter{equation}{9}
\begin{align}
P_s(r_z)&=
    \begin{cases}\label{equ_p_c_4}
    1- e^{- t}\left(1+ \frac{t}{1!}+ \frac{t^ 2}{2!}+ \dots + \frac{t^ {(n- 3)/2}}{((n- 3)/2)!}\right), & n\ \rm{is\ odd.}\\
    1- \erfc(t^ {1/2})- e^{- t}\left(\frac{t^{1/ 2}}{1/2!}+ \frac{t^{3/ 2}}{(3/ 2)!}+ \dots + \frac{t^ {(n- 3)/2}}{((n- 3)/2)!}\right), & n\ \rm{is\ even.}\\
    \end{cases}
\end{align}
\hrulefill
\vspace*{4pt}
\end{figure*}

\subsection{Asymptotic analysis}

An asymptotic error bound is given for $P \rightarrow \infty$. 
For lattice codes $\mathcal C$, this is equivalent to code rate $R \to \infty$. For $P \rightarrow \infty$, the decodable region approximately becomes an $n$-dimensional cylinder with radius $r$, where $r$ is $r_c$ or $r_e$. Similar to \eqref{equ_p_c_2}, the asymptotic probability of word error $P_{e,asym}$ is:
\begin{equation} \label{equ_asym_1}
    P_{e,asym}= 1- \int_{-\infty}^{\infty} \frac{1}{\sqrt{2 \pi \sigma^ 2}} e^{- \frac{z_1^2}{2 \sigma^2}} P_s(r)  d z_1
\end{equation}
where $P_s(r)$ is given by \eqref{equ_p_c_2_Ps} with radius $r$, and is independent of $z_1$. So we are able to separate $P_s(r)$ from the integral over $z_1$. Then  $P_{e,asym}= 1- P_s(r)$ and given $P \rightarrow \infty$ it can be found in closed form \eqref{equ_p_c_4} with a constant radius $r$. In order to obtain this result, the order of the limit and the integral are exchanged; although formal justification is required, it is correct intuitively because $P_s(r)$ and $z_1$ are increasingly independent in large $P$.
Fig.~\ref{fig_asymp_P_infinite} shows the error bound alone with its asymptotic value for dimension 8 and 16 with radius $r$ being 5.4512 and 6.5552 respectively. The $r$ are chosen such that the error bound computed by \eqref{equ_sphere_bound} is around $10^{-4}$.


\begin{figure}[t]
    \centering
    \includegraphics[scale=0.44]{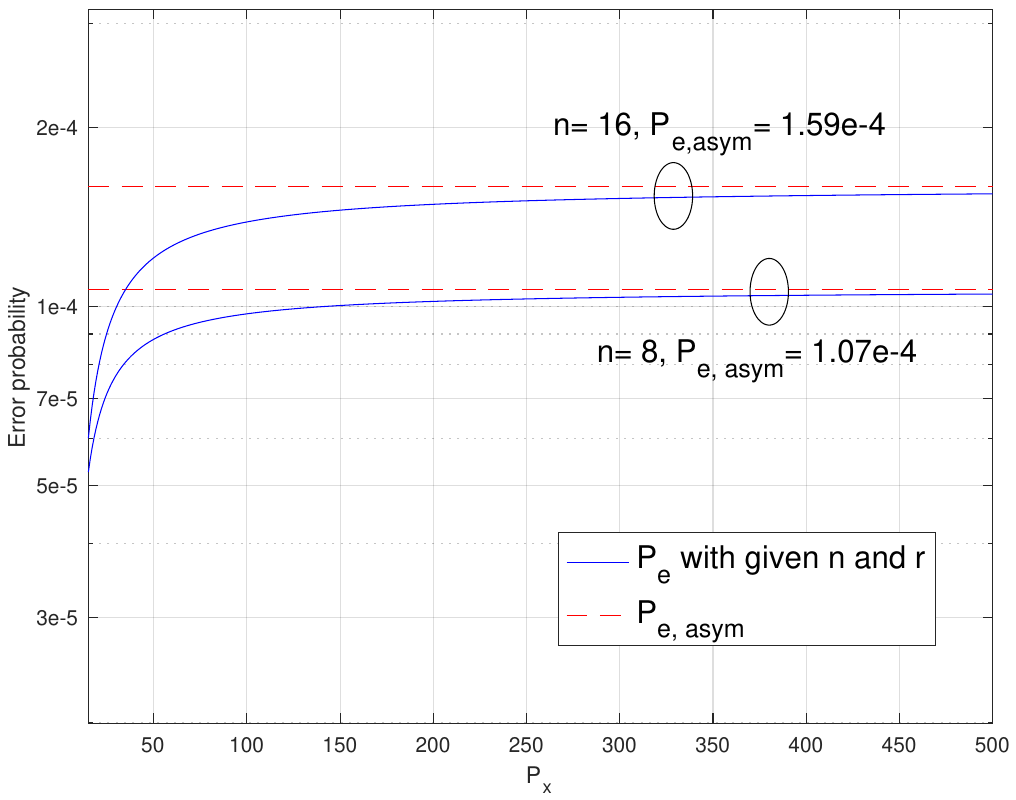}
    \caption{Asymptotic result of the error bound when $P_{\mathbf x} \rightarrow \infty$ given dimension $n= 8,\ 16$ and $r= 5.4512,\ 6.5552$ respectively. The noise variance $\sigma^2= 1$.}
    \label{fig_asymp_P_infinite}
\end{figure}


\section{Numerical Evaluation}

In this section we give a numerical evaluation of the error bound, the WER of the genie-aided exhaustive search decoder and conventional one-shot decoders over power-constrained AWGN channel. Evaluation of the following 5 cases are given in Fig.~\ref{fig_comparsion_bound}: (a) WER lower bounded using $P_{e,cover}$; (b) WER estimation using $P_{e,effc}$; (c) WER of the genie-aided exhaustive search decoder; (d) WER of the one-shot decoder with scaling factor $\alpha=1$; (e) WER of the one-shot decoder with the MMSE scaling factor $\alpha_{MMSE}$. For case (c), the $\alpha$ search range is $(0.5, 1.5)$ with step 0.01. Because E8 and BW16 lattices have known covering radius \cite[Ch.~4, p.~121, 130]{conway1993sphere}, E8 and BW16 lattice codes are considered. Hypercube shaping is applied as shaping scheme and the average message power is approximated using second moment of the shaping lattice. 
The dotted lines indicate $P_{e,cover}$, the lower bound on the WER by considering the covering sphere. The lower bound is valid for any lattice codes with given $r_c$ due to the property $\mathcal{V} \subset \mathcal{S}_c$. 
Compared with the one-shot decoder using $\alpha_{MMSE}$, the genie-aided exhaustive search decoder obtained 0.5 dB gain for E8 lattice code and 0.4 dB gain for BW16 lattice code at WER of $10^{-4}$. The estimate of WER, $P_{e,effc}$ from Subsection~\ref{sec_III_3}, shows an accurate estimate of E8 and BW16 lattice codes using the genie-aided exhaustive search decoder, which matches the hypothesis we gave in that subsection.

\begin{figure}[t]
    \centering
    \includegraphics[scale=0.44]{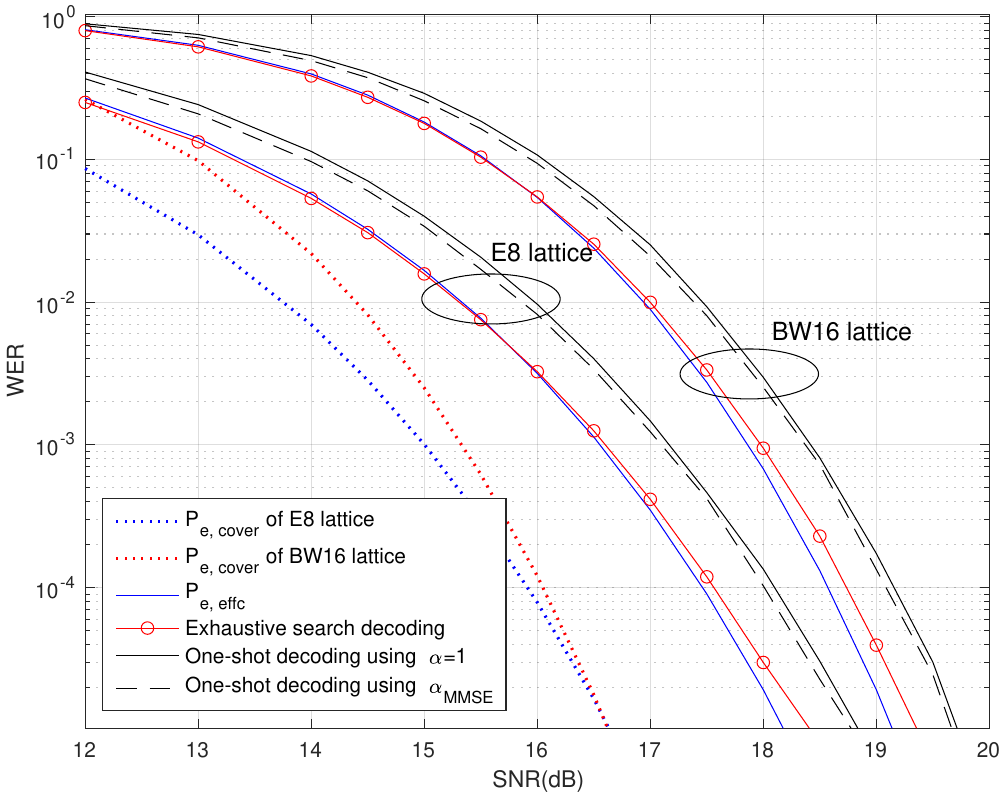}
    \caption{Numerical evaluation using E8 and BW16 lattice over power-constrained AWGN channel with the code rate being 2 and 2.25 respectively.}
    \label{fig_comparsion_bound}
\end{figure}

As a preliminary result, the genie is implemented by embedding CRC bits into lattice integer vectors. We give a simplified decoder targeted at practical implementation. Instead of an exhaustive search, the decoder selects three $\alpha$ candidates. Fig.~\ref{fig_polar_sim} shows a numerical result using $n= 128$ polar code lattice with soft cancellation (SC) decoding over the power-constrained AWGN channel. The lattice design follows \cite{ludwiniananda2021design}. With hypercube shaping, the code has 223 bits in total, 219 information bits and 4 bits of CRC parity. The resulting code rate is 1.74. Even though there is 0.078 dB SNR penalty due to the CRC bits, the lattice decoder can achieve lower WER, about 0.1 dB gain, than one-shot decoding if the decoder is allowed to re-try decoding. We hypothesize that this 0.1 dB gain can be further increased if more than 3 retries are allowed.

\begin{figure}[t]
    \centering
    \includegraphics[scale=0.44]{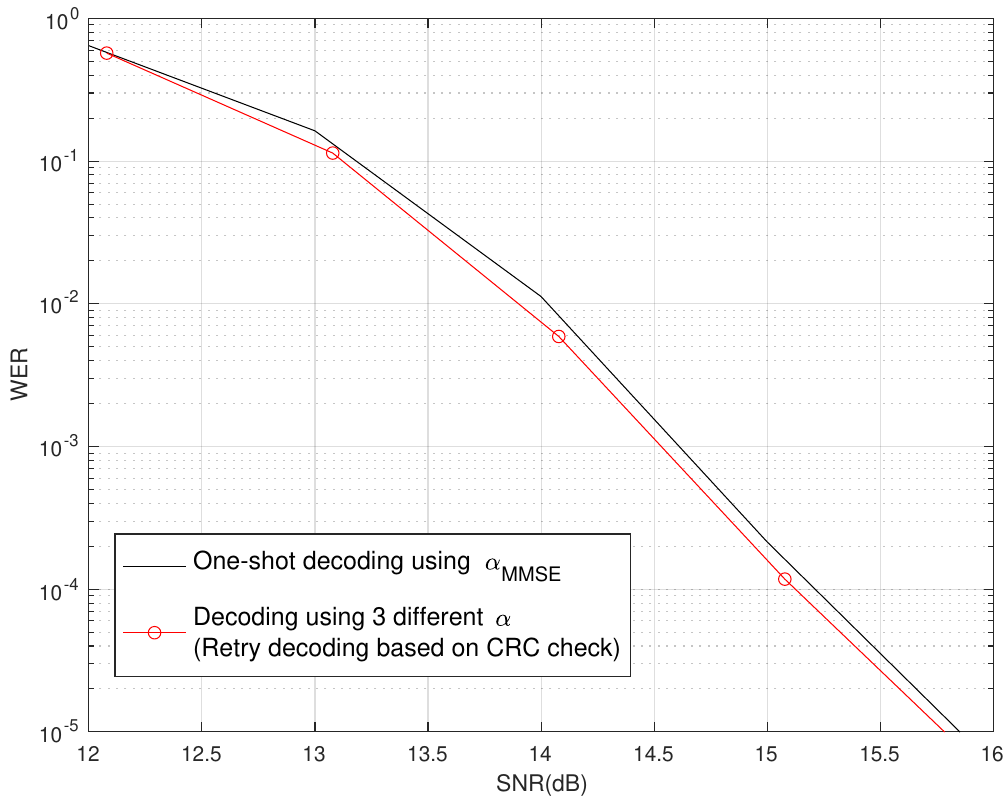}
    \caption{WER for $n=128$ polar lattice code with SC decoding over power-constrained AWGN channel, where the genie is implemented by a CRC.}
    \label{fig_polar_sim}
\end{figure}

\section{Conclusion}

In this paper, a genie-aided exhaustive search decoder was described that can decode message by searching through all $\alpha \in \mathbb{R}$. The lower bound and estimate on the WER for this decoder was given by considering the covering sphere and the effective sphere, respectively. The estimate was shown to be accurate for E8 and BW16 lattice codes. Compared with conventional one-shot decoding using $\alpha_{MMSE}$, the genie-aided decoder can obtain 0.5 dB for E8 lattice code and 0.4 dB for BW16 lattice code at WER of $10^{-4}$. 
A CRC-embedded lattice code design was considered to implement the genie. Instead of performing an exhaustive search, the decoder may select a finite number of reliable $\alpha$'s to achieve lower WER while keeping reasonable complexity. Preliminary results show that the decoder can achieve lower WER even after the SNR penalty for the CRC is included. Optimization of the embedded CRC and the strategy for selecting $\alpha$ in re-try decoding is considered to achieve lower WER at decoder.

	\bibliographystyle{ieeetr}
	\bibliography{Reference}

\end{document}